\documentclass{llncs}

\usepackage{graphicx}
\usepackage[utf8]{inputenc}
\usepackage{xcolor}
\usepackage{cite}

\newcommand{\fmed}{\mathsf{med}}
\newcommand{\fconn}{\mathsf{conn}}

\title{Lower Bounds for Leaf Rank of Leaf Powers}
\author{Svein Høgemo}
\institute{University of Bergen}
\date{February 2024}

\begin{document}

\maketitle

\begin{abstract}
Leaf powers and $k$-leaf powers have been studied for over 20 years, but there are still several aspects of this graph class that are poorly understood. One such aspect is the \emph{leaf rank} of leaf powers, i.e. the smallest number $k$ such that a graph $G$ is a $k$-leaf power. Computing the leaf rank of leaf powers has proved a hard task, and furthermore, results about the asymptotic growth of the leaf rank as a function of the number of vertices in the graph have been few and far between. We present an infinite family of rooted directed path graphs that are leaf powers, and prove that they have leaf rank exponential in the number of vertices (utilizing a type of subtree model first presented by Rautenbach [Some remarks about leaf roots. Discrete mathematics, 2006]). This answers an open question by Brandstädt et al. [Rooted directed path graphs are leaf powers. Discrete mathematics, 2010].
\end{abstract}

\section{Introduction}

A graph $G$ is a $k$-\emph{leaf power} if there is a tree $T$ such that $G$ is isomorphic to the subgraph of $T^k$ induced by its leaves. $T$ is referred to as the $k$-\emph{leaf root} of $G$. The original motivation for studying leaf powers comes from computational biology, particularly the problem of reconstructing phylogenetic trees -- if $T$ is interpreted as the ``tree of life'', then $G$ constitutes a simplified model of relationships between known species, where those species within distance $k$ in $T$ are deemed ``closely related'' and become neighbors in $G$, and those that have larger distance are deemed ``not closely related''. Checking if a graph is a $k$-leaf power for some $k$ is thereby analogous to the task of fitting an evolutionary tree to these simplified relationships. $k$-leaf powers were first introduced by Nishimura, Ragde and Thilikos in 2000 \cite{NRT02}, although the connection between powers of trees and the task of (re)constructing phylogenetic trees was explored by several authors at the time \cite{LKJ00}. Since then, the class of leaf powers -- all graphs that are $k$-leaf powers for some $k$ -- has become a well-studied graph class in its own right. A survey on leaf powers published in a recent anthology on algorithmic graph theory \cite{RLB21} gives a more or less up-to-date introduction to the most important results on this graph class.\\

Two closely related questions concern the characterization of $k$-leaf powers for some constant $k$; and the problem of computing the leaf rank (the smallest integer $k$ such that $G$ is a $k$-leaf power) of a leaf power $G$. The first problem has been addressed by several authors, most notably by Lafond \cite{Laf21}, who announced an algorithm for recognizing $k$-leaf powers that runs in polynomial time for each fixed $k$ (though, admittedly, with a runtime that depends highly on $k$). Complete characterizations in terms of forbidden subgraphs are, on the other hand, known only for 2- and 3-leaf powers~\cite{DGHN06} and partially for 4-leaf powers \cite{Rau06} (see also \cite{BLS08}). The second problem seems even harder. A few graph classes have bounded leaf rank, for example block graphs or squares of trees. The only subclass of leaf powers with unbounded leaf rank, for which leaf rank is shown to be easy to compute, is the chordal cographs (also known as the trivially perfect graphs); this was shown very recently by Le and Rosenke \cite{LR23}.\\

Though deciding the exact leaf rank of a leaf power seems hard, the asymptotic growth of the leaf rank as a function of the number of vertices has shown to be at most linear for most subclasses of leaf powers \cite{BHMW10} (also implicit in \cite{BHTV22}, see further below). This could lead one to conjecture at most linear -- or at least polynomial -- growth on the leaf rank of any leaf power. In this paper, we show that this is not the case. In particular, we show that there exists an infinite family of leaf powers $\{R_m \mid m\geq 3\}$ that have leaf rank proportional to $2^{\frac{n}{4}}$, where $n$ is the number of vertices.\\

The broader problem of recognizing leaf powers has been addressed more recently: Leaf powers, being induced subgraphs of powers of trees, are strongly chordal, as first noted in \cite{BL06}. Nevries and Rosenke \cite{NR16} find a forbidden structure in the clique arrangements of leaf powers, and find the seven forbidden strongly chordal graphs exhibiting this structure. Lafond \cite{Laf17} furthermore finds an infinite family of strongly chordal graphs that are not leaf powers, and shows that deciding if a chordal graph contains one of these graphs as an induced subgraph is NP-complete. Jaffke et al. \cite{JKST19} point out that leaf powers have mim-width 1, a trait shared with several other classes of intersection graphs. Mengel's~\cite{Men18} observation that strongly chordal graphs can have unbounded (linear in the number of vertices) mim-width suggests that the gap between leaf powers and strongly chordal graphs is quite big \cite{Jaf20}.\\

It has also been observed \cite{BR10} that leaf powers are exactly the induced subgraphs of powers of trees (so-called \emph{Steiner powers} \cite{LKJ00}). This, and the observation that $k$-leaf powers without true twins are induced subgraphs of $(k-2)$-powers of trees, forms the basis for the algorithms to recognize 5- and 6-leaf powers \cite{CK07,Duc19}, which until Lafond's breakthrough result \cite{Laf21} were the state of the art in recognizing $k$-leaf powers. One peculiar interpretation of our result is therefore that there exist induced subgraphs of powers of trees whose smallest tree powers that contain them are exponentially bigger than themselves.\\

In another direction, Bergougnoux et al. \cite{BHTV22} look at subclasses of leaf powers admitting leaf roots with simple structure: In particular, they show that the leaf powers admitting leaf roots that are subdivided caterpillars are exactly the \emph{co-threshold tolerance} graphs, a graph class lying between interval graphs and tolerance graphs (\!\!\cite{MRT88}, see Figure \ref{fig:classes}). The leaf roots constructed in \cite{BHTV22} had rational weights; however, it is not hard to see that they can be modified into $k$-leaf roots for some $k\leq 2n$. Interestingly, this shows that there is a big difference between caterpillar-shaped leaf roots and caterpillar-shaped RS models (defined in Section 2): As we will see, the graphs considered in this paper have RS models that are caterpillars, but exponential leaf rank.\\


The rest of the paper is organized as follows: In Section 2 we develop basic terminology regarding leaf powers, chordal graphs and subtree models. In Section 3 we show how each graph $R_n$ is built and show that these graphs are leaf powers, in particular rooted directed path graphs. In Section 4 we prove the main result, that $R_n$ has exponential leaf rank for every $n$. In the end, we provide a brief discussion on possible upper bounds on the leaf rank of leaf powers.

\section{Basic Notions}

We use standard graph theory notation. All trees are unrooted unless stated otherwise.

In this paper, we will assume that all trees we work with have at least three leaves, and therefore there exists at least one node of degree at least 3.

Some more specialized notions follow here:

\begin{definition}[Caterpillar]
A caterpillar is a tree in which every internal node lies on a single path. This path is called the \emph{spine} of the caterpillar.
\end{definition}

\begin{definition}[Connector]
For a leaf $v$ in a tree $T$, there is one unique node with degree at least 3 that has minimum distance to $v$. We call this node the \emph{connector} of $v$, or $\fconn(v)$. 
\end{definition}

\begin{definition}[$k$-leaf power, $k$-leaf root, leaf rank]
For some positive integer $k$, a graph $G$ is a $k$-leaf power if there exists a tree $T$ and a bijection $\tau$ from $V(G)$ to $L(T)$, the set of leaves of $T$, such that any two vertices $u,v$ are neighbors in $G$ if and only if $\tau(u)$ and $\tau(v)$ have distance at most $k$ in $T$. $T$ is called a $k$-leaf root of $G$. The leaf rank of $G$, $lrank(G)$, is the smallest value $k$ such that $G$ is a $k$-leaf power, or $\infty$ if $G$ is not a leaf power.
\end{definition}

\begin{definition}[Leaf power]
A graph $G$ is a leaf power if there exists a positive integer $k$ for which $G$ is a $k$-leaf power.
\end{definition}

\begin{definition}[Leaf span]
Given a graph class $\mathcal{F}$, the \emph{leaf span} of $\mathcal{F}$, $ls_\mathcal{F}$, is a function on the positive integers that, for each $n$, outputs the smallest $k$ such that every graph in $\mathcal{F}$ on $n$ vertices has a $k$-leaf root. Clearly, this definition only makes sense if $\mathcal{F}$ is the class of leaf powers, or a subclass thereof. Alternatively, one can define $ls_\mathcal{F}(n) = \infty$ if $\mathcal{F}$ contains a graph on $n$ vertices which is not a leaf power.
\end{definition}

In our case, we will only look at leaf powers, so in our case, the leaf span is well defined regardless.

Leaf powers are known to be \emph{chordal graphs} \cite{Rau06}, graphs with no induced cycles of four or more vertices. A famous theorem by Gavril \cite{Gav74} says that the chordal graphs are the intersection graphs of subtrees of a tree, i.e. the graphs admitting a \emph{subtree model}:

\begin{definition}[Subtree model]
Given a graph $G$, a \emph{subtree model} of $G$ is a pair $(T,\mathcal{S})$, where $T$ is a tree and $\mathcal{S} = \{S_v \mid v\in V(G)\}$ is a collection of connected subtrees of $T$ with the property that for any two vertices $u,v\in V(G)$, the subtrees $S_u$ and $S_v$ have non-empty intersection if and only if $uv\in E(G)$.
\end{definition}

\begin{definition}[Cover]
Let $G$ be a chordal graph and $(T,\mathcal{S})$ a subtree model of $G$. For a node $x\in V(T)$, the \emph{cover} of $x$, $V_{G,T}(x)$ (subscripts may be omitted), is defined as the set of vertices in $G$ whose subtrees in $T$ include $x$: $V_{G,T}(x) = \{v\in V(G) \mid x\in S_v\}$.
\end{definition}

The cover of any node must be a clique in $G$.

It is a well-known fact that subtrees of a tree have the Helly property (see e.g. \cite{Gol04}). Therefore, given a chordal graph $G$ and a subtree model $(T,\mathcal{S})$, we define the following subtrees:

\begin{definition}[Clique subtree]\label{obs:helly}
Given $G$ and$(T,\mathcal{S})$ as above, for every maximal clique $C\subseteq V(G)$, the \emph{clique subtree} $S_T(C) := \bigcap_{v\in C}S_v$ is non-empty (we can omit the subscript $T$ if the tree is obvious from context).
\end{definition}

It is therefore clear that every maximal clique of $G$ is the cover of some node in $T$.  Furthermore, for any two maximal cliques $C\neq C'$, $S_T(C)\cap S_T(C') = \emptyset$.\\

We give an alternative characterization of leaf powers here, that we will use to prove our result:

\begin{definition}[Radial Subtree model]
Given a graph $G$, a \emph{radial subtree model} (henceforth called \emph{RS model}) of $G$ is a subtree model $(T,\mathcal{S})$, where for each $v\in V(G)$ there exists a node $c_v\in V(T)$ (the \emph{center}) and integer $r_v\geq 0$ (the \emph{radius}) such that $S_v = T[\{u\in V(T)\mid dist(u,c_v)\leq r_v\}]$. In other words, $S_v$ is spanned by exactly the nodes in $T$ having distance at most $r_v$ from $c_v$. Each $S_v$ is called a \emph{radial subtree}.
\end{definition}

RS models are a special case of the much more general \emph{NeST} (Neighborhood Subtree Tolerance) models, introduced by Bibelnieks and Dearing in \cite{BD93}. NeST models are more complicated, involving trees embedded in the plane with rational distances, as well as tolerances on each vertex. We will therefore not define them here. In any case, if one removes the tolerances, the graphs admitting the resulting ``NeS models'' \cite{BHTV22} are, again, exactly the leaf powers \cite{BHMW10}. NeST graphs thus generalize leaf powers in much the same way that tolerance graphs generalize interval graphs (see Figure \ref{fig:classes}).\\

The following lemma is implicit in Rautenbach (\!\!\cite{Rau06}, Lemma 1) as a proof that leaf powers are chordal -- though RS models were not explicitly defined in that paper. We will repeat the proof here, since its contrapositive (stated below as Corollary \ref{cor:neslowerbound}) is crucial for our proof that $R_n$ has high leaf rank.

\begin{lemma}\label{lemma:nesequal}
If a graph $G$ admits a $k$-leaf root, then it admits a RS model where $\max_{v\in V(G)}r_v \leq k$.
\end{lemma}
\begin{proof}
Given a $k$-leaf root $(T,\tau)$ of $G$, we make a RS model of $G$ by:
\begin{itemize}
\item Subdividing every edge in $T$ once.
\item Setting $c_v := \tau(v)$ and $r_v := k$ for every $v\in V(G)$.
\end{itemize}
Two subtrees $S_v,S_u$ intersect iff $dist(v,u)\leq 2k$; in other words, iff $u$ and $v$ had distance at most $k$ before subdivision of the edges.
\qed\end{proof}

\begin{corollary}\label{cor:neslowerbound}
Let $G$ be a leaf power. If there is some integer $r$ such that every RS model of $G$ contains a subtree with radius at least $r$, then $G$ is not a $k$-leaf power for any $k<r$.
\end{corollary}

The radial subtrees constructed in this proof are all centered on leaves and have the same radius. However, the definition of RS models is more general, so we must prove that the implication holds in the other direction as well:

\begin{lemma}
Let $G$ be a graph. If $G$ admits an RS model $(T,\mathcal{S})$, then $G$ is a leaf power.
\end{lemma}
\begin{proof}
Let $k$ be the maximum radius among the subtrees in $\mathcal{S}$. For each $v\in V(G)$, we add a new leaf to $T$ that is fastened to $c_v$ with a path of length $k+1-r_v$, and let $\tau(v)$ point to this new leaf. Afterwards, as long as $T$ contains a leaf $x$ that is not one of these new leaves, delete the path from $x$ to $\fconn(x)$. Now, it is evident that two subtrees $S_u,s_v$ overlap if and only if $dist(\tau(u),\tau(v)) \leq 2k+2$. In other words, $(T,\tau)$ is a $(2k+2)$-leaf root of $G$.
\qed\end{proof}

\begin{figure}[ht!]
    \centering
    \includegraphics[width=0.5\linewidth]{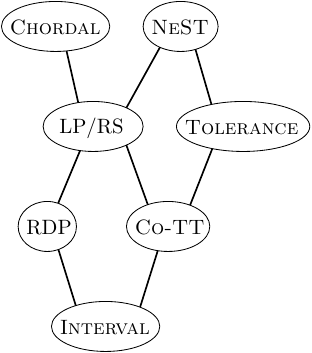}
    \caption{A Hasse diagram of inclusions between leaf powers and some related graph classes. (Abbreviations: LP=Leaf Powers; RDP=Rooted Directed Path graphs; Co-TT=Co-Threshold Tolerance graphs; RS=Graphs with RS models; NeST=Neighborhood Subtree Tolerance graphs.) All inclusions are strict and all non-inclusions are between incomparable graph classes. For more information, see \cite{BHTV22,BD93,BHMW10,GMT84,MRT88}.}
    \label{fig:classes}
\end{figure}

\section{Construction of $R_n$}

The graphs with high leaf rank that we construct are \emph{rooted directed path graphs}:

\begin{definition}[Rooted directed path graph]
A graph $G$ is a rooted directed path graph (RDP graph) if it admits an intersection model consisting of paths in an arborescence (a DAG in the form of a rooted tree where every edge points away from the root).
\end{definition}

\begin{theorem}[\!\!\cite{BHMW10}, Theorem 5]
RDP graphs are leaf powers.
\end{theorem}

The leaf roots shown to exist by Brandstädt et al. in \cite{BHMW10} had, in the worst case, $k$ exponential in $n$. They left it as an open question whether the leaf span of RDP graphs actually is significantly smaller.

Here we show that it is not: specifically, for every $n \geq 3$, there is a RDP graph $R_n$ with $4n$ vertices that has leaf rank proportional to $2^n$. In other words, if $RDP$ is the class of RDP graphs and $LP$ the class of leaf powers, then $ls_{LP} \geq ls_{RDP} = 2^{\Omega(n)}$. This is the first result that shows that the leaf span of leaf powers is non-polynomial.


For some $n\geq 3$, the graph $R_n$ has $4n$ vertices: $V(R_n) = \bigcup_{i=1}^n \{a_i,b_i,c_i,d_i\}$. We define $E(R_n)$ through its maximal cliques: The family of maximal cliques is $$\mathcal{C}(R_n) = \{C_i \mid 1\leq i\leq n\}\cup\{C'_i \mid 1\leq i\leq n-1\}$$ where $$C_i = \{a_i,b_i,c_i,d_i\}$$ and $$C'_i = \{a_j\mid i\leq j\leq n\}\cup\{b_i,b_{i+1},c_i\}$$

\begin{remark}
For any $n$, $R_n$ is a rooted directed path graph.
\end{remark}
\begin{proof}
Let our arborescence be a rooted caterpillar $T$ with spine $x_1,x_2,\ldots,x_n$ and one leaf $y_i$ fastened to each $x_i$. The root is $x_1$. Each $a_i$ corresponds to the path from $x_1$ to $y_i$; each $b_i$ corresponds to the path from $x_{i-1}$ to $y_i$ (except $b_1$, whose path starts at $x_1$); each $c_i$ corresponds to the path from $x_i$ to $y_i$; and each $d_i$ corresponds to the path consisting only of $y_i$. We can easily check that, for each $1\leq i\leq n$, $V(y_i)=C_i$; and for each $1\leq i\leq n-1$, $V(x_i) = C'_i$; and there can be no other maximal cliques since every maximal clique must be the cover of some node in $T$. (see also Figure \ref{fig:rdp} for a visual representation of the paths.)
\qed\end{proof}

\begin{figure}[ht!]
    \centering
    \includegraphics[width=0.9\linewidth]{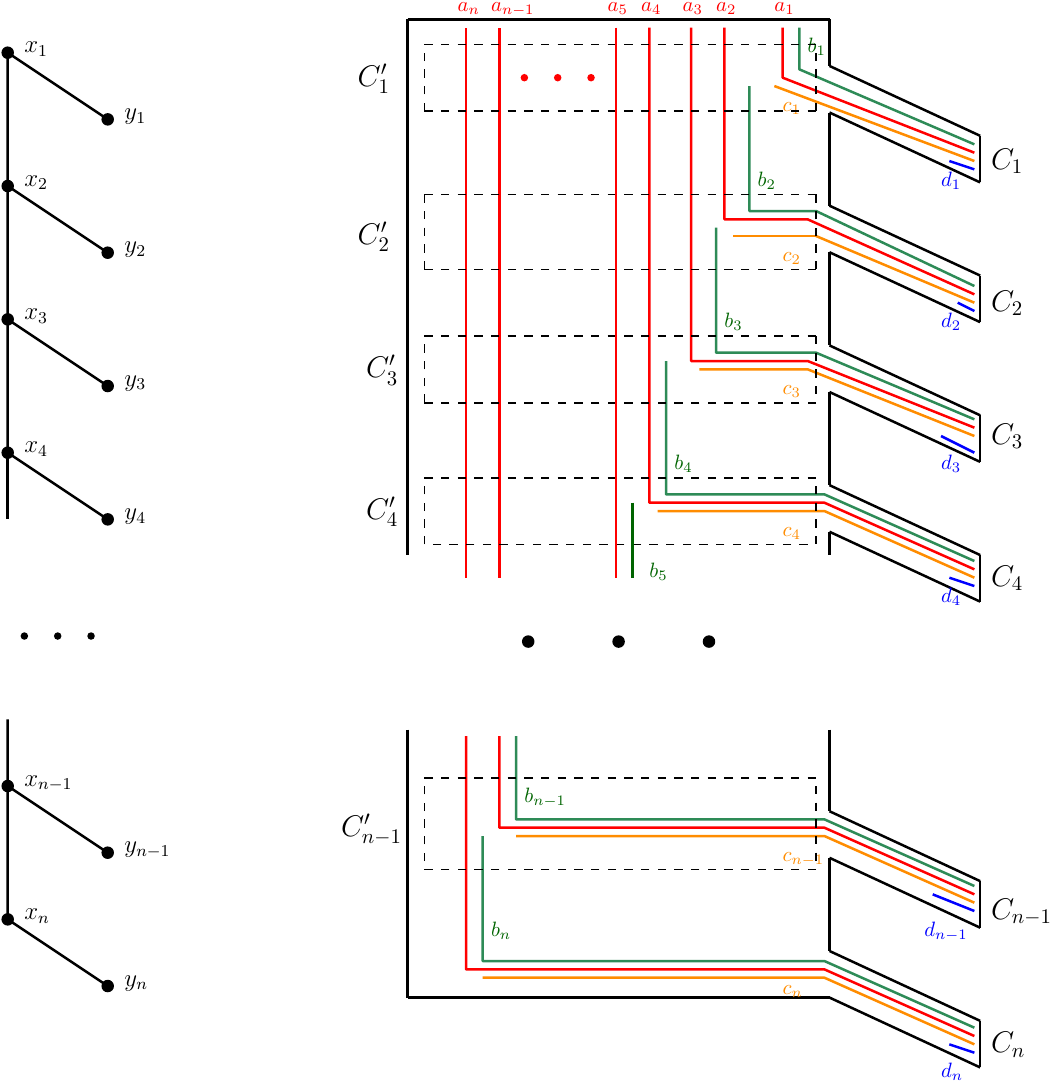}
    \caption{A rooted directed path model of $R_n$. The arborescence $T$ is a caterpillar; on the right, the tree has been fattened into a box diagram so we can see all the paths.}
    \label{fig:rdp}
\end{figure}

The construction of $R_n$ shows the ingredients we need in order to prove that some graphs have exponential leaf rank. The $b_i$'s form an induced path in $R_n$, forcing a linear topology on any subtree model of $R_n$. In other words, any subtree model, and specifically any RS model of $R_n$ must have the overall shape of a caterpillar, where the spine contains $C'_1,\ldots,C'_{n-1}$ and the hairs contain $C_1,\ldots,C_n$. The $a_i$'s have a large neighborhood, which in turn give their subtrees in any RS model a large diameter. The $d_i$'s give each $a_i$ a private neighbor, while the $c_i$'s and $b_i$'s together force every hair to branch off the spine at different points.

\section{The Graph Class $\{R_n \mid n\geq 3\}$ has Exponential Leaf Rank}

In order to prove that the aforementioned construction leads to high leaf rank, we must formalize the intuitions given earlier, and explicitly show that every RS model of $R_n$ must have a subtree with big radius. To be able to do so, we need quite a bit of new infrastructure regarding subtree models.

The first piece is a simple, but very useful lemma:

\begin{lemma}\label{lem:intersectpath}
Let $G$ be a chordal graph and $(T,\mathcal{S})$ a subtree model of $G$. Let $P$ be any path in $G$ with endpoints $u,v$, and let $x_u$ and $x_v$ be two arbitrary nodes in $S_u$ and $S_v$, respectively. For any node $x\in p(x_u,x_v)$, $V(x)\cap P \neq \emptyset$.
\end{lemma}
\begin{proof}
Assume towards a contradiction that there is a node $x_0\in p(x_u,x_v)$ whose cover does not intersect with $P$. Clearly, $x_u,x_v \neq x_0$, and therefore $x_0$ separates $x_u$ and $x_v$. We enumerate the vertices in $P$ $(p_1,p_2,\ldots,p_k)$ where $p_1=u$ and $p_k=v$. Since every subtree in a subtree model must be connected, every $S_{p_i}$ must be contained in one component of $T\setminus x_0$. Also, since $p_i$ and $p_{i+1}$ are neighbors, their respective subtrees intersect and must therefore be contained in the same component. But this leads to a contradiction, since $S_u$ and $S_v$ are located in different components of $T\setminus x_0$ (namely, the ones containing $x_u$ and $x_v$ respectively).
\qed\end{proof}

\begin{definition}[Connecting Path]
Given two disjoint subtrees of a tree $S,S'\subseteq T$, we define a \emph{connecting path} from $S$ to $S'$, denoted $p(S,S')$, as the minimal subgraph $P$ of $T$ (i.e. a path) such that $S\cup S'\cup P$ is connected. Note that $p(S,S')$ contains one node from each of $S$ and $S'$.
\end{definition}

\begin{lemma}\label{lem:cliqueatleaf}
Let $G$, $(T,\mathcal{S})$ and recall the definition of clique subtrees. Given three cliques $C,C',C''\in \mathcal{C}(G)$, if $S_T(C')$ intersects $p(S_T(C),S_T(C''))$, then $C'$ is a separator in $G$. 
\end{lemma}
\begin{proof}
Since $C,C',C''$ are maximal cliques, there are vertices $v\in C\setminus (C\cup C'')',v''\in C''\setminus (C\cup C')$. By Lemma \ref{lem:intersectpath}, every path from $c$ to $c''$ must contain a vertex in $C'$. Since, by definition, $c$ and $c''$ are not neighbors, $C'$ is a separator.
\qed\end{proof}

The above lemma is useful for us because of its contrapositive. Specifically, one can easily verify that in $R_n$, none of the cliques $C_1,\ldots,C_n$ are separators. This means that in any subtree model $(T,\mathcal{S})$ of $R_n$, the subtree $S_T(C_i)$ does not intersect the connecting path between any two other cliques. In other words, for each $1\leq i\leq n$, the clique subtree $S_T(C_i)$ is situated at a leaf of $T$.

\begin{definition}[Median]
Given a tree $T$ and three nodes $u,v,w\in V(T)$, the \emph{median} of the nodes $\fmed(u,v,w)$ is the unique node $m$ that lies on all three paths $p(u,v)$, $p(u,w)$ and $p(v,w)$. It is easy to see that $m$ is equal to one of the nodes (say, $v$) iff $v$ is on $p(u,w)$; otherwise, it separates $u,v,w$ in $T$ (and consequently, has degree at least 3).
\end{definition}

\begin{figure}[ht!]
    \centering
    \includegraphics[width=0.35\linewidth]{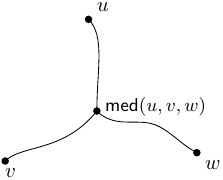}
    \caption{The median of three nodes in a tree.}
    \label{fig:median}
\end{figure}

Now we get to the meat of the proof. We will need the following definitions: Given $R_n$ and any RS model $(T,\mathcal{S})$, we note the following \emph{branch points} in $T$: Let $m_1$ and $m_n$ be the endpoints of $p(S(C_1),S(C_n))$, and for every $1<i<n$, let $m_i$ be the common node between $p(S(C_1),S(C_i))$, $p(S(C_i),S(C_n))$ and $p(S(C_1),S(C_n))$. Also, $m_i$ is the median $\fmed(m_1,m_n,s_i)$ where $s_i$ is the endpoint of $p(S(C_1),S(C_i))$ (or $p(S(C_m),S(C_i))$) in $S(C_i)$. From Lemma \ref{lem:cliqueatleaf} we know that none of $m_1$, $m_n$ and $s_i$ lie on the path between the two others; therefore $m_i$ separates $S(C_1)$, $S(C_i)$ and $S(C_n)$. By definition, every $m_i$ is on $p(m_1,m_n)$.

Take note of the nodes $m_1,\ldots,m_n$ and $s_2,\ldots,s_{n-1}$; these will all be used later on. It is worth to note that since $s_i\in S(C_i)$, the cover of $s_i$ is obviously equal to $C_i$.

We now prove a series of lemmas, concluding with Theorem \ref{thm:highleafrank}, showing that the leaf rank of $R_n$ is exponential in $n$. We will assume that $(T,\mathcal{S})$ is an RS model of $R_n$ containing the branching points mentioned above.

\begin{figure}[ht!]
    \centering
    \includegraphics[width=0.7\linewidth]{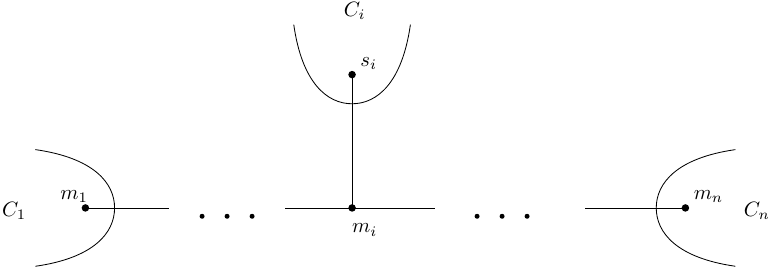}
    \caption{The branching point $m_i$ in a subtree model of $R_n$, and the three cliques it separates.}
    \label{fig:branchpoint}
\end{figure}

\begin{lemma}\label{lem:medianincludes}
For every $1<i<n$, $V(m_i)$ is equal to the union of $\{a_j \mid i\leq j\leq n\}$, $b_i$ and at least one of $c_i$ and $b_{i+1}$.
\end{lemma}
\begin{proof}
From the definition, $m_i$ separates the three subtrees $S(C_1)$, $S(C_i)$ and $S(C_n)$, represented by the three nodes $m_1$, $s_i$ and $m_n$ respectively. This means that for each of the three cliques, at least one of their vertices are not in $V(m_i)$.

We start by showing that $a_j,b_j\notin V(m_i)$ for any $j<i$: Consider the path $P=(c_i,b_{i+1},b_{i+2},\ldots,b_n)$ in $R_n$. Since $c_i\in C_i$ and $b_n\in C_n$, and $m_i$ is on the path in $T$ between those two cliques, by Lemma \ref{lem:intersectpath}, $V(m_i)$ contains one of the vertices in $P$. But none of these are adjacent to $a_j$ or $b_j$, therefore none of these can be in $V(m_i)$. Furthermore, as $\bigcup_{j=1}^{i-1}\{a_j,b_j\}$ induce a connected subgraph of $R_n$, all their respective subtrees must lie in the same component of $T\setminus m_i$; namely, the one containing $m_1$.

Next, we show that $a_i\in V(m_i)$. This is easily done by applying Lemma \ref{lem:intersectpath} to the path $(a_1,a_i)$ in $R_n$ and noting that $m_i$ is on the path $p(m_1,s_i)$ in $T$. Since we have established that $a_1\notin V(m_i)$, $a_i$ must be in $V(m_i)$.

Now we show $b_j,c_j,d_j\notin V(m_i)$ for any $i+2\leq j\leq n$, but at least one of $c_i$ and $b_{i+1}$ is. Consider the path $P$ from before. We know at least one vertex in $P$ is in $V(m_i)$, but since $c_i$ and $b_{i+1}$ are the only ones adjacent to $a_i$, they are the only ones that can be in $V(m_i)$. Furthermore, since $\bigcup_{j=i+1}^n \{b_j,c_j,d_j\}$ induce a connected subgraph of $R_n$, all of their respective subtrees must be located in the same component of $T\setminus m_i$; namely, the one containing $m_n$.

Next, we show that $a_j\in V(m_i)$ for every $i<j\leq n$. We have established that the subtrees $S_{a_1}$ and $S_{d_j}$ do not contain $m_i$, and furthermore, they are located in different components of $T\setminus m_i$. Taking the node $s_j\in S(C_j)$, we see that $m_i\in p(m_1,s_j)$. Therefore, we can apply Lemma \ref{lem:intersectpath} to the path $(a_1,a_j,d_j)$ and conclude that $a_j\in V(m_i)$.

Finally, we show $d_i\notin V(m_i)$. This is easily deduced by noting that $d_i$ and $a_n$ are not adjacent, and $a_n\in V(m_i)$.
\qed\end{proof}

\begin{lemma}\label{lem:allonpath}
None of the nodes $m_1,m_2,\ldots,m_n$ are equal. Furthermore, the path $p(m_1,m_n)$ visits all of these nodes in that order.
\end{lemma}
\begin{proof}
The first claim follows straight from Lemma \ref{lem:medianincludes} by noting that the cover of each branching point is unique. Also, it follows straight from the definition that every $m_i$ is on $p(m_1,m_n)$. For the last claim, we prove the following, equivalent formulation: For any $1\leq r<s<t\leq n$, $m_s$ lies on $p(m_r,m_t)$.

From the previous statement, it is clear that these three nodes lie on a single path, and therefore one of them lies in the middle. However, we see that $V(m_r)$ is not a separator (in $G$) of $V(m_s)\setminus V(m_r)$ and $V(m_t)\setminus V(m_r)$. By Lemma \ref{lem:cliqueatleaf}, $m_r$ cannot lie on $p(m_s,m_t)$. The same argument applies to $V(m_t)$; thus the only remaining choice is that $V(m_s)$ lies on $p(m_r,m_t)$.
\qed\end{proof}

\begin{lemma}\label{lem:increasingdistance}
For any $2<i<n$, $dist(m_i,m_{i+1}) > dist(m_2,m_i)$.
\end{lemma}
\begin{proof}
Recall that $(T,\mathcal{S})$ is an RS model; therefore, for any $v\in V(R_n)$, $S_v$ is characterized by a center $c_v$ and radius $r_v$.

Look then at the vertex $a_i$ for some $2<i<n$. From Lemma \ref{lem:medianincludes}, we know that $S_{a_i}$ contains $m_2$ and (by definition) $s_i$, but not $m_{i+1}$. Also, by Lemma \ref{lem:allonpath}, $m_i$ separates those three nodes. Given the node $c_{a_i}$, we therefore know that $dist(c_{a_i},m_{i+1}) > \max(dist(c_{a_i},m_2),dist(c_{a_i},s_i))$.

Now, since $(T,\mathcal{S})$ is an arbitrary RS model, we do not know where in $T$ the node $c_{a_i}$ is situated, but we will employ two cases, based on which component of $T\setminus m_i$ we find $c_{a_i}$ in.\\

\begin{figure}[ht!]
    \centering
    \includegraphics[width=0.9\linewidth]{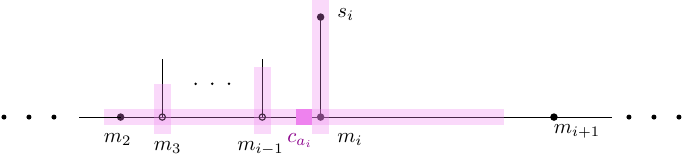}
    \caption{The purple shadow represents the radial subtree $S_{a_i}$ inside $T$, with center $c_{a_i}$ (the solid purple block). It reaches $m_2$ and $s_i$, but not $m_{i+1}$.}
    \label{fig:Nes}
\end{figure}

\emph{Case 1:} $c_{a_i}$ is not in the component of $T\setminus m_i$ containing $m_2$.

This includes the case $c_{a_i} = m_i$. In this case, we see that $$dist(m_i,m_{i+1}) \geq (dist(c_{a_i},m_{i+1})-dist(c_{a_i},m_i)) >$$ $$(dist(c_{a_i},m_2)-dist(c_{a_i},m_i)) = dist(m_i,m_2)$$
The first inequality is a strict inequality iff $c_{a_i}$ is in the component of $T\setminus m_i$ containing $m_{i+1}$; otherwise it is an equality.\\

\emph{Case 2:} $c_{a_i}$ is in the component of $T\setminus m_i$ containing $m_2$.

Now we see that $$dist(m_i,m_{i+1}) = (dist(c_{a_i},m_{i+1})-dist(c_{a_i},m_i)) >$$ $$(dist(c_{a_i},s_i)-dist(c_{a_i},m_i)) = dist(m_i,s_i)$$
(This corresponds to the scenario in Figure \ref{fig:Nes}.)

To complete the proof, we look at the center of another vertex, namely $c_{a_{i+1}}$. From Lemma \ref{lem:medianincludes}, we know that $S_{a_{i+1}}$ contains $m_2$ and $m_{i+1}$, but not $s_i$. Since $dist(m_i,m_{i+1}) > dist(m_i,s_i)$, $c_{a_{i+1}}$ must be placed in the component of $T\setminus m_i$ containing $m_{i+1}$. But now $$dist(m_i,s_i) = (dist(c_{a_{i+1}},s_i)-dist(c_{a_{i+1}},m_i)) >$$ $$(dist(c_{a_{i+1}},m_2)-dist(c_{a_{i+1}},m_i)) = dist(m_i,m_2)$$
Now we have $dist(m_i,m_{i+1}) > dist(m_i,s_i) > dist(m_i,m_2)$ and the proof is complete.
\qed\end{proof}

\begin{theorem}\label{thm:highleafrank}
The leaf rank of $R_n$ is at least $2^{n-2}$.
\end{theorem}
\begin{proof}
By Lemma \ref{lem:medianincludes}, the subtree $S_{a_n}$ contains both $m_2$ and $m_n$, and therefore has diameter at least $dist(m_2,m_n)$. From Lemma \ref{lem:allonpath} we see that $dist(m_2,m_n) = dist(m_2,m_3)+dist(m_3,m_4)+\ldots+dist(m_{n-1},m_n)$. From Lemma \ref{lem:increasingdistance} and the fact that $dist(m_2,m_3) \geq 1$, we see that $dist(m_2,m_n) \geq 2^{n-1}-1$. This implies that $r_{a_n} \geq 2^{n-2}$, which by Corollary \ref{cor:neslowerbound} and the fact that $(T,\mathcal{S})$ is an arbitrary RS model, implies that $R_n$ has leaf rank at least $2^{n-2}$.
\qed\end{proof}

\begin{corollary}
Let $\mathcal{R}=\{R_m \mid m\geq 3\}$. Then, $ls_{RDP} \geq ls_\mathcal{R} = \Omega(2^{\frac{n}{4}})$.
\end{corollary}

\section{Conclusion}

We have shown that the leaf rank of leaf powers is not upper bounded by a polynomial function in the number of vertices. While such an upper bound has never been explicitly conjectured in the literature, we nevertheless believe that this result is surprising. The only previously established lower bounds for leaf rank are linear in the number of vertices \cite{BHMW10}, and, as previously noted, most graph classes that have been shown to be leaf powers have linear upper bounds on their leaf rank as well. Though the $k$-leaf roots of RDP graphs found by Brandstädt et al. in \cite{BHMW10} had $k$ exponential in the number of vertices, the authors left it as an open question to ``determine better upper bounds on their leaf rank''.\\

Single exponential upper bounds on leaf rank of leaf powers generally have not been found, and we leave it as an open question whether the leaf span of leaf powers is $2^{\Theta(n)}$. However, we will finish with the following nice observation noted by B. Bergougnoux \cite{Ber23}, that shows that recognizing leaf powers is in \textsf{NP}. This implies a not much worse upper bound on the leaf span of leaf powers:

Given a graph $G$, a positive certificate for $G$ being a leaf power consists of a candidate leaf root $(T,\tau)$, where every internal node of $T$ has degree at least 3; and a linear program that (say) maximizes the sum of weights on each edge in $T$, while fulfilling constraints that every pair of adjacent vertices in $G$ has distance at most 1 in $T$, and every pair of non-adjacent vertices in $G$ has distance higher than 1 in $T$. If the linear program is feasible, then $(T,\tau)$ is a weighted leaf root of $G$.\\

The above linear program can be solved in polynomial time, outputting a feasible solution (if one exists) with rational weights with a polynomial number of bits. Therefore, if $G$ admits a leaf root, it admits a $k$-leaf root where $k \leq 2^{n^c}$ for some (fairly small) constant $c$. This observation also implies that if recognizing $k$-leaf powers is strongly in \textsf{P} for arbitrary $k$, then computing leaf rank is also in \textsf{P}, since given a polynomial-time algorithm for recognizing $k$-leaf powers, one could compute leaf rank by way of binary search on the value of $k$. Recognizing leaf powers would also be in \textsf{P}.

\bibliography{ref}

\end{document}